\documentclass[12pt]{article}

\usepackage{amsfonts}
\usepackage{amsmath}
\usepackage{amssymb}
\usepackage{amsthm}
\usepackage{mathtools}
\usepackage{authblk}
\usepackage{appendix}
\usepackage{bm}
\usepackage{bbm}
\usepackage{xcolor}
\usepackage{gensymb}
\usepackage{natbib}
\usepackage{subcaption}
\usepackage{lineno}

\usepackage{graphicx,url}

\newtheorem{theorem}{Theorem}[section]

\usepackage[utf8]{inputenc}

\newcommand{\bE}{\textbf{E}}
\newcommand{\bI}{\textbf{I}}
\newcommand{\bJ}{\textbf{J}}
\newcommand{\bX}{\textbf{X}}
\newcommand{\bY}{\textbf{Y}}

\newcommand{\bS}{\textbf{S}}

\newcommand{\bZ}{\textbf{Z}}

\newcommand{\bU}{\textbf{U}}

\newcommand{\bH}{\textbf{H}}
\newcommand{\bW}{\textbf{W}}
\newcommand{\bL}{\textbf{L}}

\newcommand{\bA}{\textbf{A}}

\newcommand{\bB}{\textbf{B}}

\newcommand{\bD}{\textbf{D}}

\newcommand{\bP}{\textbf{P}}

\newcommand{\mX}{\mathrm{X}}
\newcommand{\mY}{\mathrm{Y}}
\newcommand{\mD}{\mathrm{D}}
\newcommand{\mP}{\mathrm{P}}
\newcommand{\mS}{\mathrm{S}}
\newcommand{\mT}{\mathrm{T}}
\newcommand{\mV}{\mathrm{V}}
\newcommand{\mB}{\mathrm{B}}
\newcommand{\mN}{\mathrm{N}}
\newcommand{\mZ}{\mathrm{Z}}

\newcommand{\var}{\text{var}}
\newcommand{\cov}{\text{cov}}

\newcommand{\bbeta}{\bm{\beta}}

\newcommand{\T}{\intercal}

\newcommand{\zero}{\bm{0}}

\graphicspath{ {./images/} }

\newcommand{\blind}{1}
\newcommand{\spacing}{1.1}

\begin{document}

\def\spacingset#1{\renewcommand{\baselinestretch}%
{#1}\small\normalsize} \spacingset{1}

\if1\blind
{
  \title{\bf Coexchangeable process modelling for uncertainty quantification in joint climate reconstruction}
   \author[1]{Lachlan Astfalck}
    \author[2,3]{Daniel Williamson}
    \author[1]{Niall Gandy}
    \author[1]{Lauren Gregoire}
    \author[1]{Ruza Ivanovic}
    \affil[1]{School of Earth and Environment, The University of Leeds, Leeds, UK}
    \affil[2]{Department of Mathematical Sciences, Exeter University, Exeter, UK}
    \affil[3]{The Alan Turing Institute, British Library, London, UK}
    
    \setcounter{Maxaffil}{0}
    \renewcommand\Affilfont{\itshape\small}
  \maketitle
} \fi

\if0\blind
{
  \bigskip
  \bigskip
  \bigskip
  \begin{center}
    {\LARGE\bf Coexchangeable process modelling for uncertainty quantification in joint climate reconstruction}
\end{center}
  \medskip
} \fi

\bigskip
\begin{abstract}
Any experiment with climate models relies on a potentially large set of spatio-temporal boundary conditions. These can represent both the initial state of the system and/or forcings driving the model output throughout the experiment. These boundary conditions are typically fixed using available reconstructions in climate modelling studies; however, in reality they are highly uncertain, that uncertainty is unquantified, and the effect on the output of the experiment can be considerable. We develop efficient quantification of these uncertainties that combines relevant data from multiple models and observations. Starting from the coexchangeability model, we develop a coexchangeable process model to capture multiple correlated spatio-temporal fields of variables. We demonstrate that further exchangeability judgements over the parameters within this representation lead to a Bayes linear analogy of a hierarchical model.  We use the framework to provide a joint reconstruction of sea-surface temperature and sea-ice concentration boundary conditions at the last glacial maximum (23--19 kya) and use it to force an ensemble of ice-sheet simulations using the FAMOUS-Ice coupled atmosphere and ice-sheet model. We demonstrate that existing boundary conditions typically used in these experiments are implausible given our uncertainties and demonstrate the impact of using more plausible boundary conditions on ice-sheet simulation.
\end{abstract}

\noindent%
\vfill

\newpage
\spacingset{\spacing}

\section{Introduction} \label{sec:intro}

Numerical experiments are vital tools to climate science. Knowledge of physical processes embedded into software can simulate events that are otherwise impossible to observe at the required spatial and temporal resolutions. Most climate simulators utilise boundary conditions to represent the non-computed physical processes on which the simulator relies. Generally, boundary conditions are fixed using a reference run or runs from existing multi-model ensembles (MMEs) that explicitly model the boundary condition's physical process. For example, to run an ice-sheet model over the last glacial maximum (LGM) requires climate variables such as temperature, and precipitation \citep{gregoire2012deglacial} that can be obtained using the Paleoclimate Model Intercomparison Project (PMIP) ensemble of simulator runs \citep{ivanovic2016transient}. MMEs can be biased and do not necessarily span the uncertainty of the physical process \citep{salter2018quantifying}. 

If modelling boundary conditions is viewed as a statistical reconstruction problem, there is a rich literature in statistics that attempts to combine data from multiple models and historical observations to infer spatio-temporal climate properties. \cite{rougier2013second} present a generalised framework, therein termed the coexchangeable model, where exchangeability judgements over an MME along with an assumption that each model in the ensemble has the same a priori covariance with the field they aim to represent lead to a simple model that can be used to estimate a true unobserved process. Second order inference via Bayes linear methods \citep{goldstein2007bayes} for this model ensures scalable reconstructions for a spatio-temporal field, and the requirement for only prior means and variances implies an easier prior modelling task.  The coexchangeable model has seen some interest in more recent research, such as in \cite{sansom2021constraining} where it is used to model emergent constraints for future climate projections. Alternative to the methodology of \cite{rougier2013second}, a line of research stemming from \cite{chandler2013exploiting} proposes a similar fully probabilistic model with individual weightings of each simulator rather than via an assumption of exchangeability. As noted in \cite{rougier2013second}, the model in \cite{chandler2013exploiting} is a special case of the coexchangeable model. Other statistical reconstructions of boundary conditions exist; however, they are either derived from simulator data alone and are limited by the span of the ensemble \citep[e.g.][]{salter2018quantifying}, or are reconstructed from large and dense datasets \citep[e.g.][]{liu2017dimension,zhang2020bayesian,sha2019bayesian} and so are inappropriate here given the relative data sparsity at the LGM. There are a family of post-processing methodologies in the climate sciences, therein termed offline data assimilation, that adjust model output with direct observations or proxy data based on ensemble methods \citep[e.g.][]{hakim2016last,steiger2014assimilation}. We show that these methodologies, as well, may be subsumed by the coexchangeable model. Online data assimilation techniques are not possible for large climate ensembles: the simulators are run on large supercomputers, are highly bespoke, and very limited access to the model architectures is given outside of the modelling groups.

Joint reconstruction of two or more physically distinct fields may, from a theoretical perspective, appear straightforward within hierarchical Bayesian frameworks such as those proposed by \cite{chandler2013exploiting}. However, scalability becomes a much bigger challenge. This is particularly so when the fields are highly dependent, as many in climate are. In this study we jointly model sea-surface temperature (SST) and sea-ice concentration (SIC) to drive a coupled atmosphere and ice-sheet model, where the dependence between the boundary conditions is critically important. The need for scalable inference and generality makes a coexchangeable approach desirable; though extensions are required to incorporate the additional assumptions of conditional exchangeability between the fields. Thus, we develop a coexchangeable process model that offers a Bayes linear analogy to the natural Bayesian hierarchical model for the problem. From simple and natural exchangeability judgements, we develop efficient, scalable inference for joint reconstruction. 

This article proceeds as follows. Section~\ref{sec:data_intro} provides background and context to the applied modelling problem. Section~\ref{sec:theory} reviews the existing statistical literature on Bayes linear statistics and exchangeability analysis. Section~\ref{sec:joint} extends the coexchangeable model for coupled processes, providing a Bayes linear analogy to the Bayesian hierarchical model for which we then present scalable inference via geometric updating. Section~\ref{sec:application} uses the developed methodology to reconstruct SST and SIC. Section~\ref{sec:influence} discusses the results of using the reconstructions of SST and SIC to simulate ice sheet-atmosphere interactions at the LGM, and Section~\ref{sec:conclusion} concludes. Code and data are provided as part of an \texttt{R} package downloadable from \if1\blind
{ \texttt{github.com/astfalckl/exanalysis}.} \fi
\if0\blind
{ \texttt{[redacted for blind review]}.} \fi

\spacingset{1}

\section{Reconstructing boundary conditions of SST and SIC} \label{sec:data_intro}

\spacingset{\spacing}

Modelling and understanding paleoclimate events is crucial in understanding potential effects of future climate change: one way to calibrate climate and ice-sheet behaviour is to simulate the past \citep{kageyama2017pmip4,harrison2015evaluation,schmidt2014using}. The last major deglaciation ($\sim$21--7 kya) is the most natural period to study, given the relatively rich source of observations on the climate and ice sheets compared to the more distant past. To begin to study deglaciation, an ice sheet needs to be grown within the model under the steady state boundary conditions. The ice-sheet is grown by forcing an ice-sheet model with the relatively stable conditions of the LGM ($\sim$23--19 kya) until convergence. Seasonal variation is known to play a significant role in this process \citep{joughin2012ice}, so the boundary conditions contain the seasonal cycles. The character of the resulting ice sheets are sensitive to these boundary conditions and so it is crucial to use an accurate reconstruction, realistic uncertainty quantification and a method for perturbing the boundary conditions under uncertainty to force the ice sheet model.

Existing reconstructions of LGM SST and SIC are predominantly derived from either paleodata syntheses or via numerical simulation. Data only reconstructions include CLIMAP \citep{climap1981seasonal}, GLAMAP \citep{sarnthein2003overview}, MARGO \citep{kucera2005multiproxy}, and more recently \cite{paul2020global}. As these global reconstructions are based solely on proxy-based paleodata, they are subject to large measurement error, biases in polar regions, incomplete spatial coverage, and poor temporal resolution. Coupled simulations of the LGM are useful to ensure spatio-temporal coverage and consistency of SST and SIC, but these can also be very different from observations in critical regions for growing ice sheets \citep{salter2018quantifying}. PMIP is the most notable experimental body that guide protocols for coupled ocean-atmosphere models to simulate the LGM climate \citep{kageyama2017pmip4,kageyama2021pmip4}. The PMIP community has produced MMEs for different phases of the project run with different generations of model and slight adjustments in inputs. It is common practice to use a PMIP simulator output to directly force ice sheet simulations \citep[e.g.][]{gregoire2016abrupt}. One study to use both paleodata and simulators is \cite{tierney2020glacial} who adjust a model with observations to provide SST reconstruction, with uncertainty, based on a single simulator. Differences in model physics typically induce more variability than perturbations to the parameters of a single model. Therefore, reconstructions based on a single model may be biased and overconfident. Our methodology allows for the first joint reconstruction of SST and SIC that coherently combines the PMIP models (we use iterations PMIP3 and PMIP4) with available proxy data to deliver boundary conditions with uncertainty quantification.

We use three sources of data to inform our reconstructions: PMIP simulations, SST proxy data, and maximum sea-ice extents. As with \cite{rougier2013second} and \cite{sansom2021constraining}, we select a single representative simulation from each modelling group that contributed to the PMIP3 and PMIP4 MMEs in order to make the assumption of prior exchangeability reasonable. We use the MARGO SST proxy data compilation where sea-bed sediment core samples were used to infer the true LGM SST, supplemented with some more recent data from \cite{benz2016last}. Note, many of the SST measurements are far from any icesheets but can still useful in constraining a global climate simulation. Finally, we use the Southern Hemisphere maximum sea-ice extent as published in \cite{gersonde2005sea}; Northern Hemisphere extents are only available for specific regions \citep[e.g.][]{de2005reconstruction,crosta1998application,xiao2015sea}, and so we use a simple estimate of the Northern Hemisphere sea-ice extent provided by subject matter experts. SST proxy data and maximum sea-ice extents are shown in Figure~\ref{fig:data_sst} in Section~\ref{sec:application} alongside the specification of the statistical model.

\spacingset{1}

\section{Existing statistical methodology} \label{sec:theory}

\spacingset{\spacing}

SST and SIC are dependent quantities: warmer waters will support less sea-ice and vice-versa. The physics that govern the relationship between SST and SIC is represented by a series of partial differential equations, the structure and parametrisations of which change across the simulators and with reality. For example, certain simulators are predisposed to supporting more or less sea-ice at a given SST. The dependencies between SST and SIC are naturally modelled as an emergent property of the model and captured within a hierarchical framework using conditional probability statements. However, in climate modelling, specification of the probability distributions is not obvious and computation for large climate models is prohibitive. An alternative view treats expectation, rather than probability, as the primitive quantity \citep{de1975theory}; probabilities are then the expectations of indicator functions for events. This motivates second order approaches such as Bayes linear methods \citep{goldstein2007bayes}, where inference concerns expectations and variances directly rather than as a by-product of probabilistic inference. Model specification thus only concerns specifying expectation and variance, rather than the full probability distributions. \cite{rougier2013second} show the advantages of second order specification for climate modelling and further show that based on certain judgements of exchangeability, efficient methods for inference of MMEs may be formed. Our application requires extending the theory of \cite{rougier2013second} to exchangeable processes, that is, second order hierarchical models. First, we review the existing methodology: Section~\ref{sec:bayeslin} provides a short introduction to Bayes linear theory, Section~\ref{sec:bl_exchangeable} defines second order exchangeability, and Section~\ref{sec:rougier} presents the coexchangeable model of \cite{rougier2013second}.

\spacingset{1}
\subsection{Bayes Linear Theory} \label{sec:bayeslin} 
\spacingset{\spacing}

Under the Bayes linear paradigm, beliefs on random quantities are described via expectations and variances and are then \textit{adjusted} by data. The belief specifications define an inner product space in which the random quantities live; the inner product space is analogous to a probability distribution in probabilistic Bayesian analysis. Consider random quantity, $\mX$, with observations $\mX^i$, defined on the Hilbert space $\mathcal{X}$ endowed with inner product $\langle \mX, \mY \rangle = \mathbb{E}[\mX^\T \mY]$. Denote by $\mD$ the collection of observed data, $\mD = (\mX^1, \dots, \mX^m)$, as a concatenated vector of the observations. In general, the $\mX^i$ do not necessarily have the same length or a-priori belief specifications and can represent multiple sources of data that inform $\mX$. The random quantity of interest, $\mX$, may be multivariate or univariate, in which case its inner product is simply $\mathbb{E}[\mX \mY]$. Later, $\mX$ will denote the unknown spatio-temporal field of SST, and the $\mX^i$, the observed simulator outputs over which we will assume exchangeability. We write the adjusted expectation in terms of $\mX$ and $\mD$ as $\mathbb{E}_\mD[\mX]$, that is, the expectation of our beliefs $\mX$ adjusted by data $\mD$. Adjusted expectation is defined as the element in the subspace $\mathcal{D} = \text{span}[1, \mD]$ that minimises $\|\mX - \mathbb{E}_\mD[\mX]\|$ and has solution
\begin{equation} \label{eqn:adj_exp}
  \mathbb{E}_\mD[\mX] = \mathbb{E}[\mX] + \cov[\mX,\mD] \var[\mD]^\dagger(\mD - \mathbb{E}[\mD])
\end{equation}
where $\var[\mD]^\dagger$ is any pseudo-inverse of $\var[\mD]$, most commonly the Moore-Penrose inverse. Equation~(\ref{eqn:adj_exp}) describes the orthogonal projection of each element of our beliefs, $\mX$, onto $\mathcal{D}$. The adjusted variance, $\var_\mD[\mX]$, is the outer product $\mathbb{E}\left[\left(\mX - \mathbb{E}_\mD[\mX]\right)\left(\mX - \mathbb{E}_\mD[\mX]\right)^\T\right]$ given $\mathbb{E}_\mD[\mX]$ in (\ref{eqn:adj_exp}) and is
\begin{equation} \label{eqn:adj_var}
  \var_\mD[\mX] = \var[\mX] - \cov[\mX,\mD] \var[\mD]^\dagger \cov[\mD,\mX].
\end{equation}
Derivations of (\ref{eqn:adj_exp}) and (\ref{eqn:adj_var}) are found in Sections 12.4--12.5 of \cite{goldstein2007bayes}.

\spacingset{1}
\subsection{Bayes linear analysis of exchangeable data} \label{sec:bl_exchangeable}
\spacingset{\spacing}

Exchangeability for a sequence of random quantities within a probabilistic Bayesian analysis represents a simple a priori symmetry judgement that implies that any finite sub-collection of quantities within the sequence have the same distribution. Second order exchangeability for such a sequence is an analogue for judgements when expectation is primitive, and implies that any finite sub-collection of quantities share the same joint inner product space, i.e. prior expectation and variance. If data, $\mD$, are second order exchangeable we may, according to the second order representation theorem \citep{goldstein1986exchangeable}, write
\begin{equation} \label{eqn:exchangeability}
 \mathrm{X}^i = \mathcal{M} + \mathcal{R}^i
\end{equation}
for a common mean $\mathcal{M}$ and uncorrelated residuals $\mathcal{R}^i$. Second order exchangeability implies that all observed $\mX^i$ are of the same length, $\cov[\mX^i, \mX^{i'}] = \cov[\mX, \mX']$ and $\var[\mX^i] = \var[\mX]$, $\forall i,i'$. For second order exchangeable sequences there is predictive sufficiency for updating beliefs on $\mX$ by only updating $\mathcal{M}$ by the data $\mD$. Geometrically, this means that given $\mathcal{M}$, $\mD$ and $\mX$ are orthogonal and thus uncorrelated; these results are established in \cite{goldstein1986exchangeable}. Further, the sample mean, $\bar{\mX} = \frac{1}{m} \sum_{i=1}^m \mX^i$, is Bayes linear sufficient for updating beliefs on $\mathcal{M}$, and consequently on $\mX$. Second order exchangeability affords two main advantages: first, belief specifications are simplified; and second, sufficiency of the sample mean makes inference independent of the number of samples $m$.

\spacingset{1}
\subsection{Exchangeability analysis for multi-model ensembles} \label{sec:rougier}
\spacingset{\spacing}

\cite{rougier2013second} leverage second order exchangeability to describe a Bayes linear approach for modelling MME's. In what follows, we explicitly define multivariate quantities, represented by the bold font. Let $\mathbb{X} \coloneqq \{\bX^1, \dots, \bX^m \}$ be a collection of $q$-dimensional outputs from the $m$ simulators that form the MME; $\bX^*$, the true unobserved process that the simulators aim to model; and $\bZ_\bX$, the noisy and incomplete observation of $\bX^*$. The model requires only two assumptions: first, that the $\bX^i$ are second order exchangeable and, second, that the $\bX^i$ are \textit{coexchangeable} with the truth $\bX^*$, implying $\cov[\bX^*, \bX^i] = \Sigma, \; \forall i$. As in (\ref{eqn:exchangeability}), exchangeability within $\mathbb{X}$ implies
\begin{equation} \label{eqn:X}
  \bX^i = \mathcal{M}_{\bX} + \mathcal{R}^i_{\bX}, \hspace{5mm} i = 1, \dots, m,
\end{equation}
where $\mathcal{M}_{\bX}$ is a shared mean term, the $\mathcal{R}^i_{\bX}$ are the zero-mean, uncorrelated residuals of each simulator with common variance, and the $\mathcal{M}_{\bX}$ and $\mathcal{R}^i_{\bX}$ are uncorrelated. Coexchangeability between $\bX^*$ and $\mathbb{X}$ implies sufficiency of $\mathcal{M}_{\bX}$ for $\bX^*$. This allows us to write 
\begin{equation} \label{eqn:Y}
  \bX^* = \bA \mathcal{M}_{\bX} + \bU_\bX
\end{equation}
where $\bA$ is a known matrix, and $\bU_\bX$ represents the ensemble discrepancy that is uncorrelated with $\mathcal{M}_{\bX}$ and the $\mathcal{R}^i_{\bX}$. The data, $\bZ_\bX$ are modelled as
\begin{equation} \label{eqn:Z}
  \bZ_\bX = \bH_\bX \bX^* + \bW_\bX
\end{equation}
for measurement error $\bW_\bX$, and known incidence matrix $\bH_\bX$. The statistical model defined by (\ref{eqn:X})--(\ref{eqn:Z}) will hereafter by referred to as the \textit{coexchangeable model}; a graphical representation is provided in Figure~\ref{fig:flowchart}a at the end of Section~\ref{sec:joint}. Inference for this model makes use of Bayes linear sufficiency of the ensemble mean, $\bar{\bX} = \frac{1}{m}\sum_{i=1}^m \bX^i$, for updating by $\mathbb{X}$ and is therefore very efficient. Updated beliefs on $\bX^*$ is done in two stages: first the update of our beliefs by the ensemble, and second by the data; the inferential procedure is outlined in the supplementary material. Climate post-processing routines, as in \cite{hakim2016last} and \cite{steiger2014assimilation}, are subsumed by the coexchangeable model, and simply describe (\ref{eqn:Z}) with beliefs $\mathbb{E}[\bX^*] = \bar{\bX}$ and $\var[\bX^*] = \bP$. Calculations that rely on ensemble methods approximate $\bP$ by the empirical covariance matrix of the ensemble, or some linear transformation thereof \citep[see definitions in][]{whitaker2002ensemble,snyder2022optimal}. Our contributions, that build on the coexchangeable model, may also be considered as similar developments to such post-processing routines popular in the climate sciences.


\spacingset{1}
\section{The coexchangeable process model} \label{sec:joint}
\spacingset{\spacing}

\subsection{Exchangeable processes} \label{sec:exchangeable_processes}

Consider each simulator as producing output pairs of $q$-dimensional fields $(\bX^i, \bY^i)$ with the corresponding fields in reality denoted $(\bX^*, \bY^*)$, for which we have partial observations $\bZ_\bX$ and $\bZ_\bY$, made with some error. In what follows we define SST by $\bX$ and SIC by $\bY$. Changes in SST and SIC are driven by complex physical relationships. We capture structural dependencies between random quantities $(\bX, \bY)$, by first only imposing the coexchangeable model for $\bX^i, \bX^*$, and then considering exchangeability judgements over the processes $\bY^i$, given $\bX^i$, and coexchangeability of $\bY^*$, given $\bX^*$. Whilst it is enticing to consider the coexchangeable model over both fields simultaneously, it is deficient here in two ways. First, due to the dependence of $\bY^i$ on $\bX^i$, the assumption of second order exchangeability between the $\bY^i$ is violated. Second, the natural way to construct the model is to define the process of $\bY^i$ given $\bX^i$ parametrised by some $\bbeta^i$ as in a hierarchical model. Inference on the $\bbeta^i$ 
is of scientific interest; however, the coexchangeable model does not allow for this. The model sophistication required to handle the hierarchical structure requires nuanced judgements of conditional second order exchangeability. In this section, we simply state the exchangeable process model, and defer detailed discussion of the conditional exchangeability judgements to Section~\ref{sec:repeated_obs}.

For any single simulator we model $\mathbb{E}[\bY^{i}] = \mathcal{M}(\bX^{i}; \bm{\bbeta}^i)$ as a process of $\bX^{i}$ parametrised by some $\bm{\bbeta}^i$, specific to the $i$th simulator. The mean function $\mathcal{M}(\bX^i; \bbeta^i)$ may represent any relationship between $\bX^i$ and $\bbeta^i$ and, to infer $\bbeta^i$ from $\bY^i$, we need only specify a joint inner product space in which they reside. We set $\mathcal{M}(\bX^i; \bbeta^i) = \phi(\bX^i) \bbeta^i$ where $\phi(\cdot)$ maps $\bX^i$ to some specified family of basis functions. In our application, we use a monotonic decreasing spline basis for $\phi(\cdot)$ to reflect the property that warmer SSTs will tend to lead to less SICs. We may then write
\begin{equation} \label{eqn:exchangeability_Y}
  \bY^{i} = \phi(\bX^i) \bbeta^i + \mathcal{R}^{i}_{\bY}
\end{equation} 
where $\mathcal{R}^{i}_{\bY}$ is some associated residual term. A simple example to consider is where $\bX^{i}$ and $\bY^{i}$ represent temporal observations at a single location, $\phi(\bX^i) = \bX^i$ and $\bbeta^i$ is a scalar. This simply models $\bX^{i}$ and $\bY^{i}$ as scalar multiples of each other, where the multiplier is simulation dependent. Such examples are common when modelling emergent constraints in climate modelling. Note, the choice to model $\mathcal{M}(\bX^i, \bbeta^i)$ linear in $\bbeta^i$ is not required by the theory but it allows for natural specification of the inner product space and leads to sufficiency arguments, discussed below, that aid computation. To complete the analogy with a Bayesian hierarchical model, we impose second order exchangeability over $\bbeta^1, \bbeta^2, \dots$ so that
\begin{equation} \label{eqn:exchangeability_beta}
  \bbeta^i = \mathcal{M}_{\bbeta} + \mathcal{R}^i_{\bbeta},
\end{equation}
with expectation $\mathbb{E}[\bbeta^i] = \mathcal{M}_{\bbeta}$ and covariance $\cov[\bbeta^i, \bbeta^{i'}] = \var[\mathcal{M}_{\bbeta}] + \mathbf{1}_{\{i = i'\}}\var[\mathcal{R}^i_{\bbeta}]$.

Conditional exchangeability does not lead to sufficiency of the sample mean and so we must calculate the belief updates in the much larger joint inner product space. Define $\Phi_i \coloneqq \phi(\bX^i)$; equations (\ref{eqn:exchangeability_Y}) and (\ref{eqn:exchangeability_beta}) imply that the joint space of the $\bY^i$ and the $\bbeta^i$ can be formed as
\begin{equation} \label{eqn:joint1}
\spacingset{1}
  \begin{bmatrix}
    \bY^1 \\
    \vdots \\
    \bY^m \\
    \bbeta^1 \\
    \vdots \\
    \bbeta^m
  \end{bmatrix} =
  \left[\begin{array}{@{}c|c@{}}
    \begin{matrix}
    \Phi_1 & \cdots & \zero_{q \times k} \\
    \vdots & \ddots & \vdots \\
    \zero_{q \times k} & \cdots & \Phi_m
    \end{matrix}
    & \zero_{qm \times k} \\
  \hline
    \zero_{km \times km} &
    \begin{matrix}
    \bJ_{m \times 1} \otimes \bI_{k}
    \end{matrix}
  \end{array}\right]
  \begin{bmatrix}
    \bbeta^1 \\
    \vdots \\
    \bbeta^m \\
    \mathcal{M}_{\bbeta}
  \end{bmatrix} + 
  \begin{bmatrix}
    \mathcal{R}^1_\bY \\
    \vdots \\
    \mathcal{R}^m_\bY \\
    \mathcal{R}^1_{\bbeta} \\
    \vdots \\
    \mathcal{R}^m_{\bbeta}
  \end{bmatrix},
\end{equation}
\spacingset{\spacing}
where $\zero_{a \times b}$ is a $a \times b$ zero matrix, $\bJ_{a \times b}$ is a $a \times b$ matrix of ones, and $\bI_a$ is the $a \times a$ identity matrix. 

Calculating the belief updates using (\ref{eqn:joint1}) is not immediately obvious as the random quantities $\bbeta^i$ appear on both the left-hand side as data and the right-hand side as unknown parameters. Noting that (\ref{eqn:exchangeability_beta}) is equivalently restated as $0 = \mathcal{M}_{\bbeta} + \mathcal{R}^i_{\bbeta} - \bbeta^i$, and using the re-parametrisation of \cite{hodges1998some}, we can re-express (\ref{eqn:joint1}) as
\begin{equation} \label{eqn:hodges_reparam}
\spacingset{1}
  \begin{bmatrix}
      \bY^1 \\
      \vdots \\
      \bY^m \\
      \zero_{km \times 1}
    \end{bmatrix} =
  \left[\begin{array}{@{}c|c@{}}
    \begin{matrix}
    \Phi_1 & \cdots & \zero_{q \times k} \\
    \vdots & \ddots & \vdots \\
    \zero_{q \times k} & \cdots & \Phi_m
    \end{matrix}
    & \zero_{qm \times k} \\
  \hline
    -\bI_{km} &
    \begin{matrix}
    \bJ_{m \times 1} \otimes \bI_{k}
    \end{matrix}
  \end{array}\right]
  \begin{bmatrix}
    \bbeta^1 \\
    \vdots \\
    \bbeta^m \\
    \mathcal{M}_{\bbeta}
  \end{bmatrix} + 
  \begin{bmatrix}
    \mathcal{R}^1_\bY \\
    \vdots \\
    \mathcal{R}^m_\bY \\
    \mathcal{R}^1_{\bbeta} \\
    \vdots \\
    \mathcal{R}^m_{\bbeta}
  \end{bmatrix}
\end{equation} \spacingset{\spacing}with the familiar linear canonical form $\bY = \bX \bB + \bE$.  As we consider the joint representation, expectation is taken jointly over the $\bY^i$ and $\bbeta^i$. The adjusted expectation, $\mathbb{E}_\bY[\bB]$, and adjusted variance, $\text{var}_\bY[\bB]$, are calculated via (\ref{eqn:adj_exp}) and (\ref{eqn:adj_var}); the joint belief specifications and specific updating equations are provided in the supplementary material. We note that the specification of (\ref{eqn:joint1}) differs from the standard approach to modelling exchangeable processes in Bayes linear statistics where inference only concerns $\mathcal{M}_{\bbeta}$ by, in effect, substituting (\ref{eqn:exchangeability_beta}) into (\ref{eqn:exchangeability_Y}) \citep[see][]{goldstein1998adjusting}. By jointly modelling the $\bbeta^i$ and $\mathcal{M}_{\bbeta}$ we provide a closer analogy to Bayesian hierarchical models. 

As we are required to update our beliefs in the joint inner product space, computation of $\mathbb{E}_\bY[\bB]$ and $\text{var}_\bY[\bB]$ can be difficult. These calculations require an order $\mathcal{O}(m^3(q+k)^3)$ matrix inversion of $\var[\bY]$, where $m$ is the number of simulations in the MME, $q$ is the dimension of the climate simulation, and $k$ is the dimension of the $\bbeta^i$. Theorem~\ref{thm:proj_sufficient} shows that a Bayes linear sufficiency argument can be made that permits a smaller computational order of $\mathcal{O}(8k^3m^3)$ for $k < q$, thus enabling efficient inference. Theorem~\ref{thm:proj_sufficient} says, in effect, that the beliefs on $\bB$ may be equivalently updated by the projection of the $\bY^i$ onto the $k$-dimensional column space of $\Phi_i$. For climate models, $q$ is generally very large. For many basis designs $k \ll q$, and when groups, $i$, index MIP simulations, $m$ is generally small.

\begin{theorem} \label{thm:proj_sufficient}
Let $\hat{\bbeta} = (\hat{\bbeta}^1, \dots, \hat{\bbeta}^m)$ with $\hat{\bbeta}^i = (\Phi_i^\T \Phi_i)^\dagger \Phi_i^\T \bY^i$. Then $\hat{\bbeta}$ is Bayes linear sufficient for $\bY$ for adjusting $\bB$ if the column space of projection matrix $\mP_i = \Phi_i(\Phi_i^\T \Phi_i)^\dagger \Phi_i^\T$, $C(\mP_i)$, is invariant over $i$, i.e., $C(\mP_i) = \mathcal{C}$, $\forall i$.
\end{theorem}

\begin{proof}
Proof available in the supplementary material.
\end{proof}

Matrix $\Phi_i$ being of full rank is sufficient, but not necessary, in satisfying the condition $C(\mP_i) = \mathcal{C}$, $\forall i$ in Theorem~\ref{thm:proj_sufficient}, and most sensible basis designs will ensure this. Assuming $C(\mP_i) = \mathcal{C}$, $\forall i$, and appealing to the sufficiency of $\hat{\bbeta}$ for $\bY$ for adjusting $\bB$, we write the joint specification of the exchangeability judgements made in (\ref{eqn:hodges_reparam}) as
\begin{equation} \label{eqn:hierarchical_reparam}
\spacingset{1}
  \begin{bmatrix}
      \hat{\bbeta}^1 \\
      \vdots \\
      \hat{\bbeta}^m \\
      \zero_{km \times 1}
    \end{bmatrix} =
  \left[\begin{array}{@{}c|c@{}}
    \begin{matrix}
    \bI_{km}
    \end{matrix}
    & \zero_{km \times k} \\
  \hline
    -\bI_{km} &
    \begin{matrix}
    \bJ_{m \times 1} \otimes \bI_{k}
    \end{matrix}
  \end{array}\right]
  \begin{bmatrix}
    \bbeta^1 \\
    \vdots \\
    \bbeta^m \\
    \mathcal{M}_{\bbeta}
  \end{bmatrix} + 
  \begin{bmatrix}
    \mathcal{R}^1_{\hat{\bbeta}} \\
    \vdots \\
    \mathcal{R}^m_{\hat{\bbeta}} \\
    \mathcal{R}^1_{\bbeta} \\
    \vdots \\
    \mathcal{R}^m_{\bbeta}
  \end{bmatrix}
\end{equation}
\spacingset{\spacing}
where $\mathcal{R}^i_{\hat{\bbeta}} = (\Phi_i^\T \Phi_i)^\dagger \Phi_i^\T \mathcal{R}^i_\bY$. Calculation of $\mathbb{E}_{\hat{\bbeta}}[\bB]$ and $\text{var}_{\hat{\bbeta}}[\bB]$ follows (\ref{eqn:adj_exp}) and (\ref{eqn:adj_var}); the specific equations are given in the supplementary material. We also provide partial updating equations for $\mathbb{E}_{\hat{\bbeta}}[\mathcal{M}_{\bbeta}]$ and $\text{var}_{\hat{\bbeta}}[\mathcal{M}_{\bbeta}]$ which are calculated with computational order $\mathcal{O}(k^3 m^3)$ and may be used in equations (S4) and (S5) that update the full coexchangeable process model.

\spacingset{1}
\subsection{Repeated observations of a process} \label{sec:repeated_obs}
\spacingset{\spacing}

Assume, within each simulator output $i$, we observe a sequence $(\bX^i_1, \bY^i_1), (\bX^i_2, \bY^i_2), \dots$ of observation/covariate pairings; here, each $(\bX^i_t, \bY^i_t)$ is some $p$-dimensional sub-level process contained in $(\bX^i, \bY^i)$ and the $\bbeta^i$ are invariant to the dimension indexed by $t$. In the simple example provided above, each $(\bX^i_t, \bY^i_t)$ are pairings of scalar observations, and so $p=1$, and $\bbeta^i$ is invariant in time. Similarly, for our application, each $t$ indexes time, but the $\bX^i_t$ and $\bY^i_t$ are spatial observations within the spatio-temporal $\bX^i$ and $\bY^i$; though we could instead index space, both space and time, or some other feature of the process. The exchangeable process model implicitly requires an assumption of conditional exchangeability in (\ref{eqn:exchangeability_Y}), the second order equivalent of the exchangeability judgements used in Bayesian regression \citep[for discussion see][]{williamson2020emergent}. Specifically, conditional second order exchangeability implies that the $\bY^i_1, \bY^i_2, ...$ are second order exchangeable given some $\bX^i_t$ that is constant for all $t$. We may thus specify a mean process, $\mathbb{E}[\bY^i_t] = \mathcal{M}(\bX^i_t; \bbeta^i)$ for some known covariate $\bX^i_t$ and unknown parameter $\bbeta^i$; above, we assume the mean process to be $\mathbb{E}[\bY^i_t] = \phi(\bX^i_t) \bbeta^i$. According to the second order representation theorem,
\begin{equation} 
  \bY^{i}_t = \phi(\bX^i_t) \bbeta^i + \mathcal{R}^{i}_{\bY_t}
\end{equation} 
where $\cov[\bY^i_t, \bY^i_{t'}] = \cov[\mathcal{M}(\bX^i_t; \bbeta^i), \mathcal{M}(\bX^i_{t'}; \bbeta^i)] + \mathbf{1}_{\{t = t'\}} \var[\mathcal{R}^i_{\bY_t}]$ and we arrive at (\ref{eqn:exchangeability_Y}) by building the joint representation, that is, stacking the instances of $t$. In can be intuitive to think of the full simulator outputs $(\bX^i, \bY^i)$ as matrices with the variant and invariant dimensions (here, space and time) indexed across the rows and columns, respectively. The mathematics in Section~\ref{sec:exchangeable_processes} simply then require $\bX^i$ and $\bY^i$ to be substituted by their vectorised equivalents, $\mathrm{vec}(\bX^i)$ and $\mathrm{vec}(\bY^i)$.

\spacingset{1}
\subsection{Reality as a coexchangeable process}\label{sec:reality_process}
\spacingset{\spacing}

We now show how judgements of coexchangeability may be made to incorporate (\ref{eqn:hodges_reparam}) or (\ref{eqn:hierarchical_reparam}) into the coexchangeable model. Define the true process of $\bY^i$ that the simulators attempt to resolve as $\bY^*$. An assumption coexchangeability of $\bY^*$ and $\bY^i$ given fixed $\bX^i = \bX$ $\forall i$, and hence $\Phi_i = \Phi$ $\forall i$, is equivalent to assuming coexchangeability of $\bY^*$ and the $\bbeta^i$. Bayes linear sufficiency of $\mathcal{M}_{\bbeta}$ for the $\bY^i$ for adjusting $\bY^*$ follows. We write a model for $\bY^*$ such that
\begin{equation} \label{eqn:reality_process_Y}
\bY^* = \bA_{\bY} \mathcal{M}_{\bbeta} + \bU_{\bY},
\end{equation}
where $\bU_{\bY}$ is a model-mismatch term uncorrelated with the $\bY^i$ and the assumption of coexchangeability permits $\bA_{\bY}$ to be any matrix of suitable dimensions. The choice of $\bA_{\bY} = \phi(\bX^*)$ is obvious in our context but different choices are permissible should they make sense to the application. Finally, the data, $\bZ_\bY$ are modelled as
\begin{equation} \label{eqn:data_process_Y}
    \bZ_\bY = \bH_\bY \bY^* + \bW_\bY,
\end{equation}
for measurement error $\bW_\bY$, and known incidence matrix $\bH_\bY$.

Similar to \cite{rougier2013second}, we assume sufficiency of $\bY^*$ for $\bZ_\bY$, allowing us to update beliefs on $\bY^*$ in two stages. The assumption of sufficiency of $\bY^*$ for $\bZ_\bY$ is akin to an assumption of conditional independence 
between $\bZ_\bY$ and the $\bY^i$ in a probabilistic analysis. The two-stage update of beliefs on $\bY^*$ may thus be equivalently thought of as either the joint update by the $\bY^i$ and $\bZ_\bY$ or as sequential updates by the $\bY^i$ and then $\bZ_\bY$. As with the coexchangeable model in Section~\ref{sec:rougier}, we show an graphical representation of coexchangeable process model in Figure~\ref{fig:flowchart}b; the equations for the two-stage belief update are provided in the supplementary material.

\spacingset{1}

\begin{figure}[h!]
  \centering
  \includegraphics[width = 150mm]{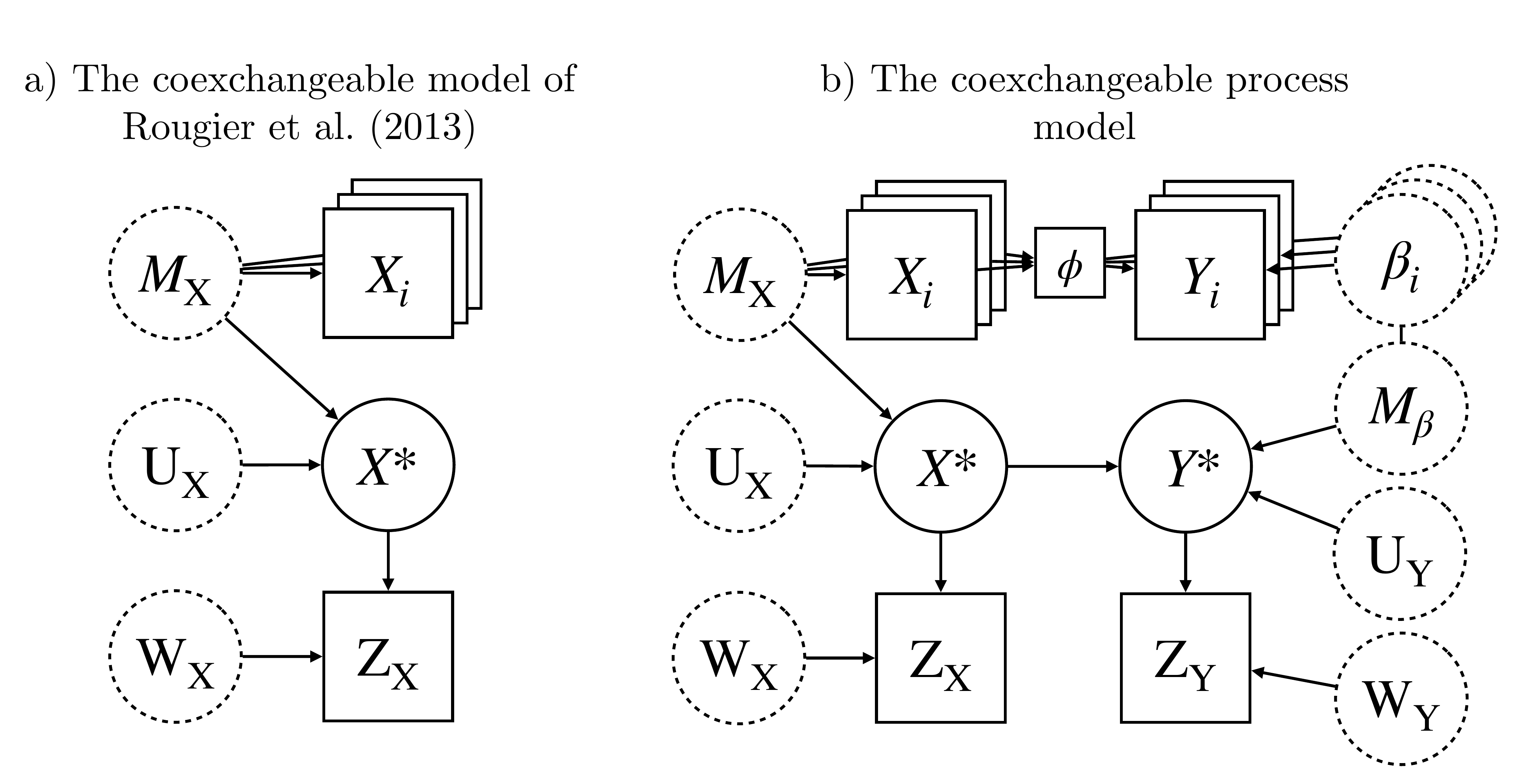}
  \caption{Graphical representation of the coexchangeable and coexchangeable process models. Boxes represent observed quantities, dashed circles represent unobserved quantities over which we make prior belief specifications, and solid circles represent unobserved quantities for which we calculate updated beliefs. Analogous to conditional independence in probabilistic models, arrows may be used to identify Bayes linear sufficiency between quantities. For simplicity, residual terms are omitted.}
  \label{fig:flowchart}
\end{figure}

\spacingset{1}

\section{Joint reconstruction of Paleo sea-surface temperature and sea-ice with PMIP3 and PMIP4} \label{sec:application}

\subsection{Paleodata}

\spacingset{\spacing}

Geological, ecological and geochemical measurements of the LGM have large associated uncertainties, and these uncertainties are further compounded by relating the measured processes into proxies for SST and SIC. We use judgements from subject matter experts to account for sampling bias and the errors in the proxy-data, which are considered to be systematic in space. In some cases we directly incorporate these judgements into the belief structure of the model, namely, expectations, variances, and covariances. In other cases, similar to \cite{rougier2022estimating}, we use `pseudo-observations' to reflect subject matter experts' judgements in sparsely observed areas.

SST data, $\bZ_{\bX}$, uses proxy measurements obtained from sea-bed sediment core samples recorded either as annual or summer means. The measurements are shown in Figure~\ref{fig:data_sst}; annual and summer means are depicted with points and triangles, respectively. There is very likely some strong observation bias in the foraminifera-based proxy measurements from the Arctic and Nordic seas, approximately north of $62^\circ$ N. In this region, very cold water and full sea-ice coverage are common; both inhospitable conditions for most foraminifera species. The non-existence of foraminifera in ocean sediments are rarely reported, since the absence of foraminifera leaves little to be analysed in ecological or geochemical studies. Thus, observations may be biased towards the warm climate events naturally present within the inter-decadal variability, when foraminifera are found. We account for this through the specification of the measurement error term in (\ref{eqn:Z}), $\bW_\bX$. Define $\mB_{\mN}$ as a set of indices that index observations from the Nordic Seas in $\bZ_\bX$. We set $\mathbb{E}[\bW_\bX]_{b} = 2$ and $\mathbb{E}[\bW_\bX]_{b^*} = 0$ for $b \in \mB_{\mN}$ and $b^* \notin \mB_{\mN}$. As the bias originates from a systematic reporting error we believe the measurement errors to be correlated. We partition $\var[\bW_\bX] = \mV_\bD(\bI + \mV_\bB)\mV_\bD$ and set $\mV_\bD$ as the diagonal matrix of the reported standard deviations in the MARGO dataset, and $\mV_{\bB[b,b']} = 1 \; \forall b, b' \in \mB_{\mN}$ where $b \neq b'$, and $0$ otherwise. This represents the weakest possible belief specification as it makes the measurement error in the Nordic seas perfectly correlated, in essence reducing the information into a single observation. Without accounting for the bias in these observations, the Nordic Seas would be too warm to support sea-ice, which is known from geological records not to be the case. The specification of $\mathbb{E}[\bW_\bX]_b = 2$ comes from subject matter experts; even still, due to the imposed correlation structure the analysis is not sensitive to this specification. Finally, we set incidence matrix in (\ref{eqn:Z}), $\bH_\bX = \bH_\bX^{\mT} \otimes \bH_\bX^{\mS}$, where $\bH_\bX^{\mT}$ calculates either the annual or summer means of $\bX^*$ as necessary, and $\bH_\bX^{\mS}$ spatially interpolates these averages from simulator's spatial grid to the data locations.

\spacingset{1}

\begin{figure}[h!]
  \centering
  \includegraphics[width = 140mm]{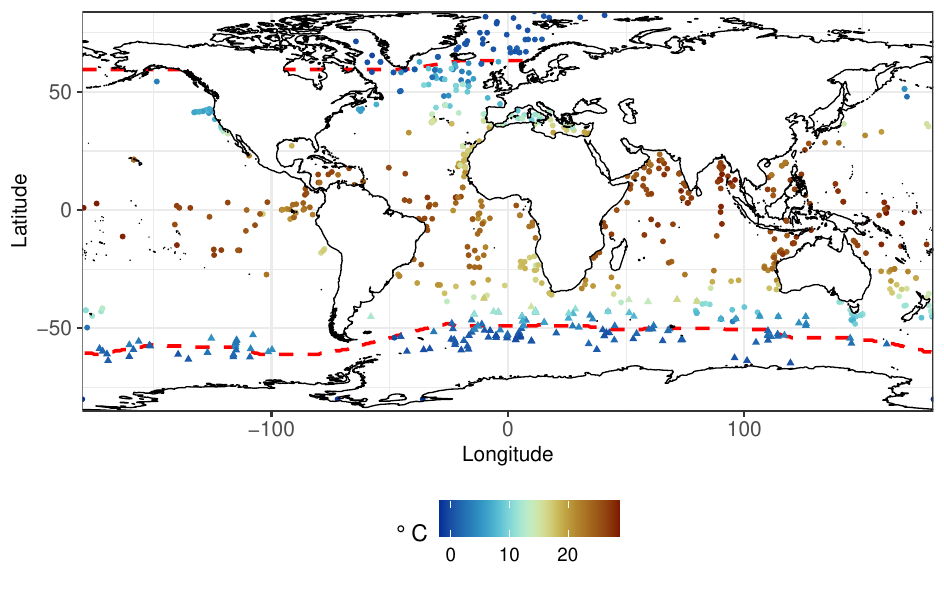}
  \caption{Proxy-based measurements of LGM SST and maximum sea-ice extents. Annual and summer mean SST's are represented by points and triangles, respectively. Sea-ice extents are represented by a red dashed line. The Southern sea-ice extent is as reported in \cite{gersonde2005sea} and the Northern sea-ice extent is provided by coauthors.}
  \label{fig:data_sst}
\end{figure}

\spacingset{\spacing}

Point-wise proxy measurements of SIC are difficult to interpret and unreliable. More robust are estimates of maximum sea-ice extent. The Northern and Southern extents used for $\bZ_{\bY}$ are shown in Figure~\ref{fig:data_sst} by the red dashed lines. $\bZ_{\bY}$ is a $4145$-dimensional vector, i.e. the same spatial resolution as the numerical models, that records a $1$ at spatial locations within the extent boundaries, and a $0$ outside. Incidence matrix, $\bH_\bY$ in (\ref{eqn:data_process_Y}) is a $p \times q$ matrix that collates, from $\bY^*$, the February SIC from the Northern hemisphere and August SIC from the Southern hemisphere. We specify measurement error, $\bW_\bY$, to be spatially correlated and certain of SIC estimates in the poles, where we are confident that there is full sea-ice coverage, and mid-latitude and equatorial regions, where we are confident there is zero sea-ice. We set $\var[\bW_\bY] = K \text{cor}[\bW_\bY] K$ where $K$ is a diagonal matrix that represents our marginal uncertainty
and $\text{cor}[\bW_\bY]$ is spatially correlated error; descriptions of $K$ and $\text{cor}[\bW_\bY]$ are given in the supplementary material.

\spacingset{1}
\subsection{Fitting the coexchangeable process model} \label{sec:model}
\spacingset{\spacing}

To fit the coexchangeable process model we first model SST via the coexchangeable model described in Section~\ref{sec:rougier} and the process of SIC given SST is modelled using the methodology developed in Section~\ref{sec:joint}. As described above, we build a MME by selecting a single representative simulation from each of the PMIP3 and PMIP4 modelling groups, with the exception of the HadCM3 model simulations where we make use of all three available PMIP4 simulations that use different ice sheet boundary conditions. The $m=13$ models selected are 
\begin{equation}
    \begin{split}
        \mathcal{S} = \{&\texttt{CNRM-3, IPSL-3, MIROC-3, MPI-3, CCSM4-3, GISS-3, MRI-3, AWI-4,} \\
        &\texttt{MIROC-4, MPI-4, HadCM3-PMIP, HadCM3-Glac1D, HadCM3-Ice6G}\}.
    \end{split}
\end{equation}
The assumption of coexchangeability is a prior judgement, and at the time of the analysis each of these models was deemed coexchangeable by the project's SMEs. Each ensemble member was projected onto the FAMOUS ocean grid with $4145$ spatial locations, and we use the 12 monthly means of SST and SIC.

\spacingset{1}
\subsubsection{Sea-surface temperature} \label{sec:sst}
\spacingset{\spacing}

Following notation in Section~\ref{sec:rougier}, define the ensemble SSTs as $\mathbb{X} = \{\bX^1, \dots, \bX^m\}$, and assume second order exchangeability within the ensemble. This leads to the representation in equation (\ref{eqn:X}), for which we require prior specification of $\var[\mathcal{M}_{\bX}]$ and $\var[\mathcal{R}^i_{\bX}]$. Coexchangeabilty of $\bX^i$ and $\bX^*$ leads to (\ref{eqn:Y}), for which we require prior specification of model mismatch terms $\mathbb{E}[\bU_\bX]$, $\var[\bU_\bX]$ and incidence matrix $\bA$. We specify $\var[\mathcal{R}_{\bX}] = \alpha^2 \var[\mathcal{M}_{\bX}]$ so that $\var[\bX] =  (1 + \alpha^2) \var[\mathcal{M}_{\bX}]$ and, following \cite{rougier2013second}, set $\mathbb{E}[\bU_\bX] = \zero$ and $\bA = \bI$.

For most climate fields, including SST, we can exploit spatio-temporal structure so that computation of adjusted beliefs is scalable. The most obvious way to do this is to assume separability through space and time so that $\var[\bX]$ and $\var[\bX^*]$, and hence $\var[\bU_\bX]$, have a Kronecker structure. For example, $\var[\bX] = \var[\bX_{\mT}] \otimes \var[\bX_{\mS}]$ where $\mT$ denotes time and $\mS$ denotes space. If we similarly partition $\var[\bU_\bX] = \var[\bU_{\bX_{\mT}}] \otimes \var[\bU_{\bX_{\mS}}]$ and equate either $\var[\bX_{\mT}] = \var[\bU_{\bX_{\mT}}]$ or $\var[\bX_{\mS}] = \var[\bU_{\bX_{\mS}}]$ computation of our adjusted beliefs is efficient. Note that this assumption is weaker than the application in \cite{rougier2013second} where it is assumed that $\kappa^2 \var[\bX] = \var[\bU]$. Here, we set $\var[\bX_{\mT}] = \var[\bU_{\bX_{\mT}}] = \bJ_{n \times n}$ as our subject matter experts believe that discrepancies between the $\bX^i$ and reality are predominantly spatial and constant in time. 

We set $\var[\bX_{\mS}]$ as the positive semi-definite matrix that minimises the distance between $\var[\bX]$ and the empirical covariance matrix of $\mathbb{X}$, $\bS_{\bX}$, under the Frobenius norm. This specification follows similar arguments to \cite{rougier2013second} where $\var[\bX] = \bS_{\bX}$, but preserves Kronecker structure in $\var[\bX]$ to allow for efficient computation. Details of this calculation are in the supplementary material. Elements of $\var[\bU_{\bX_{\mS}}]$ are defined via the $C^4$-Wendland covariance function such that
\begin{equation} \label{eqn:wendland}
    \var[\bU_{\bX_{\mS}}]_{[s,s']} = \kappa^2 \left(1 + \tau \frac{d(s,s')}{c} + \frac{\tau^2 - 1}{3} \frac{d(s,s')^2}{c^2}\right)\left(1 - \frac{d(s,s')}{c}\right)^\tau_+
\end{equation}
where $\tau \geq 6$, $c \in (0, \pi]$, $(a)_+ = \text{max}(0, a)$, and $d(i,j)$ is the geodesic distance between locations $i$ and $j$. The $C^4$-Wendland covariance function is commonly chosen so as to define a smooth process on the sphere; \citep[see, for example][]{astfalck2019emulation}. Parameters are specified as $\kappa = 1.61$, $c = 0.92$, and $\tau = 6$; these values are selected to represent the subject matter experts' beliefs as to the magnitude and correlation lengths
of $\bU_\bX$. A sensitivity analysis is provided in the supplementary material to highlight how these judgments influence inference.

The 2-stage Bayes linear update follows \cite{rougier2013second}. As with \cite{rougier2013second}, we assume the first update $\mathbb{E}_{\bar{\bX}}[\bX^*]$ is well approximated by $\mathbb{E}_{\bar{\bX}}[\bX^*] \approx \bar{\bX}$ and we calculate $\var_{\bar{\bX}}[\bX^*] = \left(\frac{\alpha^2}{m + \alpha^2}\right)\var[\mathcal{M}_{\bX}] + \var[\bU_\bX]$ so $\var_{\bar{\bX}}[\bX^*] \rightarrow \var[\bU]$ as $m \rightarrow \infty$. Here, we choose $\alpha^2 = 1$ and have $m = 13$ and so do not assume $\var_{\bar{\bX}}[\bX^*] \approx \var[\bU]$. From specifications $\var[\bU_{\bX_{\mT}}] = \var[\bX_{\mT}]$ and $\var[\bX] = 2 \var[\mathcal{M}_{\bX}]$, $\var_{\bar{\bX}}[\bX^*]$ has Kronecker structure so that $\var_{\bar{\bX}}[\bX^*] = \var_{\bar{\bX}}[\bX^*_{\mT}] \otimes \var_{\bar{\bX}}[\bX^*_{\mS}]$, where $\var_{\bar{\bX}}[\bX^*_{\mT}] = \bJ_{n \times n}$ and $\var_{\bar{\bX}}[\bX^*_{\mS}] = \frac{\var[\bX_{\mS}]}{2(m + 1)}  + \var[\bU_{\bX_{\mS}}]$. The above specifications lead to a prior predictive for $\bX^*$ 
that is warmer than our true beliefs at the poles where we are certain there was full sea ice coverage and so the SST must be $-1.92^{\degree}$C. This is problematic due to the data sparsity in the poles, and so we correct our prior by adding 10 equally longitudinally spaced  pseudo-observations at 80$^{\degree}$N and 80$^{\degree}$S, each of $-1.92^{\degree}$C. Figures~\ref{fig:sst_update}a and \ref{fig:sst_update}b plot $\mathbb{E}_{\bar{\bX}}[\bX^*]$ and marginal standard deviation of $\var_{\bar{\bX}}[\bX^*]$, respectively, for January. 

The second update is calculated by 
\begin{linenomath}
\begin{align}
  \mathbb{E}_{\bar{\bX}, \bZ_\bX}[\bX^*] &= \bar{\bX} + \var_{\bar{\bX}}[\bX^*] \bH_\bX^\T \var_{\bar{\bX}}[\bZ_{\bX}]^{-1}(\bZ_\bX - \bH_\bX \bar{\bX} - \mathbb{E}[\bW_\bX]) \text{, and}  \label{eqn:exp} \\
  \var_{\bar{\bX}, \bZ_\bX}[\bX^*] &= \bJ_{n \times n} \otimes \left( \var_{\bar{\bX}}[\bX^*_{\mS}] - \var_{\bar{\bX}}[\bX^*_{\mS}]\left(\bH_\bX^{\mS}\right)^\T \var_{\bar{\bX}}[\bZ_{\bX}]^{-1} \bH_\bX^{\mS} \var_{\bar{\bX}}[\bX^*_{\mS}] \right), \label{eqn:var}
\end{align}
\end{linenomath}
where $\var_{\bar{\bX}}[\bZ_{\bX}] = \bH_\bX \var_{\bar{\bX}}[\bX^*] \bH_\bX^\T + \var[\bW_\bX]$. Figures~\ref{fig:sst_update}c and \ref{fig:sst_update}d show these updates, also for January, respectively. Figures~\ref{fig:sst_update}e and ~\ref{fig:sst_update}f give indication of what information is gained from the data: Figure~\ref{fig:sst_update}e shows the difference the ensemble mean and $\mathbb{E}_{\bar{\bX}, \bZ_\bX}[\bX^*]$; Figure~\ref{fig:sst_update}f shows the empirical standard deviation of the ensemble, which when compared to Figure~\ref{fig:sst_update}d is shown to be more uncertain. Regions of cooling are apparent westward of large continental masses, i.e. on the North American Pacific Coast, or the African Atlantic Coast. This is indicative of models not capturing up-welling phenomena, and is a pattern previously observed in the error between models and pre-industrial simulations \citep{eyring2019taking}, as well as the cooling in the Southern Ocean, and warming the Indian and Pacific oceans. Larger uncertainty is seen in the Pacific where measurements are sparse.

\spacingset{1}

\begin{figure}[h!]
  \centering
  \includegraphics[width = 140mm]{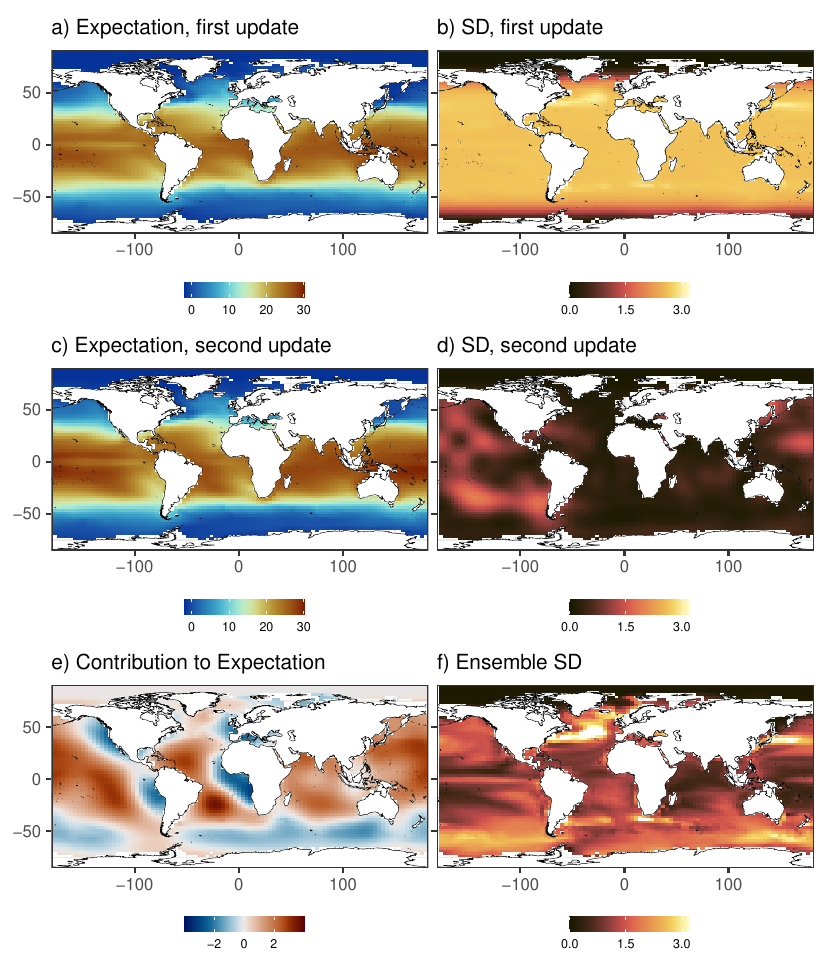}
  \caption{Adjusted beliefs of January SST: (a) expectation of SST adjusted by $\bar{\bX}$, equal to the ensemble mean; (b) marginal standard deviation of SST adjusted by $\bar{\bX}$; (c) expectation of SST adjusted by $\bar{\bX}$ and $\bZ_\bX$; (d) marginal standard deviation of SST adjusted by $\bar{\bX}$ and $\bZ_\bX$; (e) contribution of the data to our expected beliefs of SST $\mathbb{E}_{\bar{\bX}, \bZ_\bX}[\bX^*] - \bar{\bX}$; and (f), marginal standard deviation of the ensemble. All plots are shown in degrees Celsius.}
  \label{fig:sst_update}
\end{figure}

\spacingset{\spacing}

\spacingset{1}
\subsubsection{Sea-ice concentration given sea-surface temperature} \label{sec:sic}
\spacingset{\spacing}

Following notation in Section~\ref{sec:joint} we now consider the process of SIC given SST. Define the jointly-observed ensemble of SST and SIC as $(\mathbb{X}, \mathbb{Y}) = \{(\bX^1, \bY^1), \dots, (\bX^m, \bY^m)\}$ where each $\bY^i$ is a spatio-temporal vector of SIC that we model dependent on SST, $\bX^i$. The $\bY^i$ comprise 12 (i.e. monthly) $4145$-dimensional conditionally second order exchangeable spatial processes, $\bY^i_t$, where $i$ indexes the ensemble member and $t$ indexes time. We assume $\mathbb{E}[\bY^i_t] = \phi(\bX^i_t) \bbeta^{i}$ which leads to the representation in (\ref{eqn:exchangeability_Y}), for which we require specifications of $\phi(\cdot)$, $\mathbb{E}[\bbeta^i]$, $\var[\bbeta^i]$, and $\var[\mathcal{R}^i_{\bY_t}]$. We assume the $\bbeta^i$ to be second order exchangeable which leads to the representation in (\ref{eqn:exchangeability_beta}) that requires specification of $\mathbb{E}[\mathcal{M}_{\bbeta}]$, $\var[\mathcal{M}_{\bbeta}]$ and $\var[\mathcal{R}_{\bbeta}]$. Together, these exchangeability judgements lead to the representation for the $\bY^i$ in (\ref{eqn:joint1}), and equivalently (\ref{eqn:hodges_reparam}). 

For modelling SIC given SST 
we require spatial variation in the physics of the process. For example, the relationship between SIC and SST is different in the Nordic seas where sea-ice is supported at warmer SSTs than other locations. To account for this, we specify $\phi(\bX^i_t) = \Psi(\bX^i_t) \bm{\Theta}^i$ where $\Psi(\bY^i_t)$ models the behaviour of SIC and SST at individual locations using spline bases, and $\bm{\Theta}^i$ is a fixed-rank spatial basis of the spline coefficients. Note that we can model the spline coefficients individually in space, in which case $\phi(\bX^i_t) = \Psi(\bX^i_t)$, but as the size of $\bbeta^i$ then scales with the spatial resolution; this is not feasible for climate models. Here we specify $\Psi(\bX^i_t)$ with I-spline bases, a basis family commonly used for monotone functions \citep{ramsay1988monotone}. To ease computation, we project the spatial basis $\bm{\Theta}^i$ onto a principal component design calculated using a projection of the $\bY^i$ onto the column space of the $\Psi(\bX^i)$, $\hat{\Theta}^i$. Approximating the spatial coefficients using principal components restricts inference for the $\bbeta^i$ and $\mathcal{M}_{\bbeta}$ to linear combinations of $\hat{\Theta}^i$; we argue this is appropriate for modelling the mean MME process. We use a more flexible parametrisation in the model discrepancy below so that inference on $\bY^*$ given the SIC data is not restricted to linear combinations of the principal component design. Full specification of $\Psi(\bX^i_t)$ and the $\bm{\Theta}^i$ are in the supplementary material. All $\Phi_i = \phi(\bX^i)$ are full rank and so the conditions of Theorem~\ref{thm:proj_sufficient} are met and $\hat{\bbeta} = (\hat{\bbeta}^1, \dots, \hat{\bbeta}^m)$ is sufficient for $(\mathbb{X}, \mathbb{Y})$ for updating $\bB = (\bbeta^1, \dots, \bbeta^m, \mathcal{M}_{\bbeta})$. We specify $\var[\mathcal{R}^i_{\bY_t}]$ as a heteroskedastic error process; smaller variance is specified for very cold SSTs, where we are confident there is sea-ice, and warm SSTs, where we are confident there is no sea-ice, and larger variance is specified for SSTs approximately between $-1^{\degree}$C and $3^{\degree}$C where sea-ice behaviour is variable. As above we assume $\mathbb{E}[\mathcal{M}_{\bbeta}]$ is well approximated by the empirical mean of the MME members so that $\mathbb{E}[\mathcal{M}_{\bbeta}] = \frac{1}{m} \sum_{i=1}^m \hat{\bbeta}^i$; further, similar to \cite{rougier2013second}, we set $\var[\bbeta^i] = \lambda$ where $\lambda$ is a diagonal matrix of the eigenvalues calculated from the principal component decomposition described above. As before we set $\var[\mathcal{R}_{\bbeta}] = \var[\mathcal{M}_{\bbeta}]$ and so $\var[\mathcal{R}_{\bbeta}] = \var[\mathcal{M}_{\bbeta}] = \lambda/2$.

Adjusted beliefs $\mathbb{E}_{\hat{\bbeta}}[\bB]$ and $\var_{\hat{\bbeta}}[\bB]$ 
are calculated as in the supplementary material. As mentioned in Section~\ref{sec:exchangeable_processes} an advantage of specifying our model jointly over the $\bbeta^i$ and $\mathcal{M}_{\bbeta}$ is that we obtain inference for each $\bbeta^i$ \citep[as opposed to only $\mathcal{M}_{\bbeta}$ as in][]{goldstein1998adjusting}. We look at the individual fits in Figure~\ref{fig:betai_update} where we plot, at four locations, the fit ensemble members prior to the coexchangeable adjustment. Location A shows a point in the Arctic where the models largely agree. Location B shows a region that occasionally sees small concentrations of sea-ice in the models but is predominantly too warm to support full sea-ice coverage. Location C shows three MME members whose relationship of SIC given SST disagree with the rest of the ensemble.
Finally, location D shows sea-ice behaviour at the Antarctic sea-ice edge. Given the SSTs that are observed, the differing physics predominantly lead to differences in Winter sea-ice.

\spacingset{1}

\begin{figure}[h!]
\centering
 \begin{subfigure}[b]{\textwidth}
     \centering
     \includegraphics[width = 110mm]{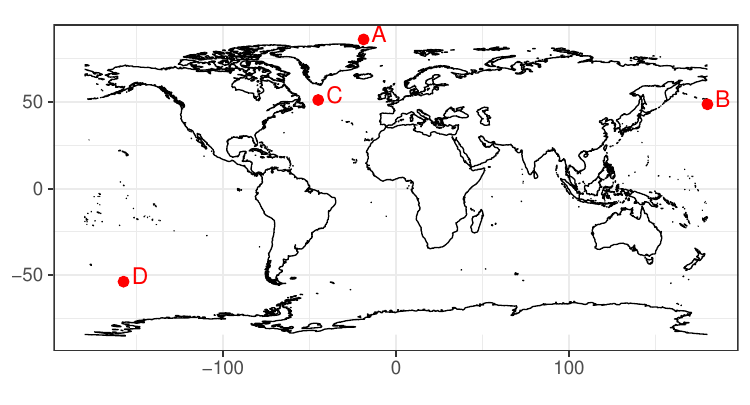}
 \end{subfigure}
 \begin{subfigure}[b]{\textwidth}
     \centering
     \includegraphics[width = 160mm]{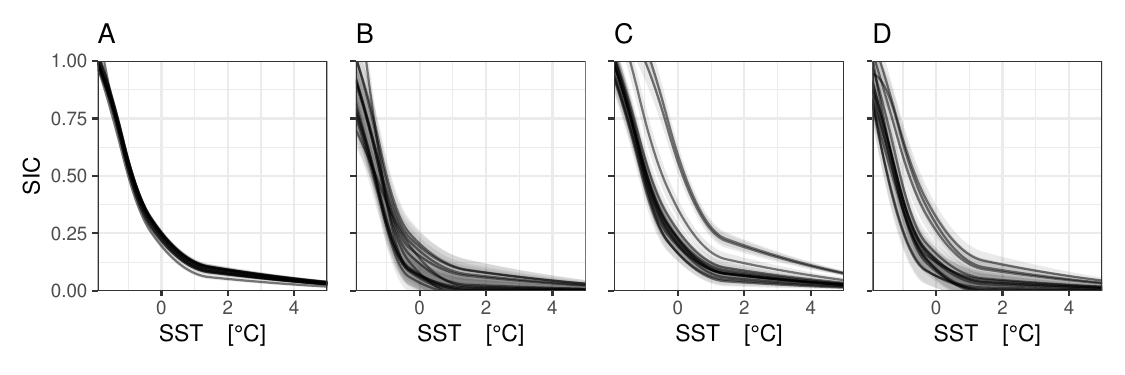}
 \end{subfigure}
 \caption{Spline fits, parameterised by $\bbeta^i$, of SIC given SST to each ensemble member prior to the coexchangeable adjustment. Shaded regions denote $\pm 2$ standard deviations.}
 \label{fig:betai_update}
\end{figure}

\spacingset{\spacing}

Coexchangeability of $\bY^*$ and $\bbeta^i$ leads to (\ref{eqn:reality_process_Y}) for which specification of $A(\bX^*)$, $\mathbb{E}[\bU_\bY]$ and $\var[\bU_\bY]$ is required. We set $A(\bX^*) = \phi(\bX^*)$ and model discrepancy so that $\bU_\bY = \Psi(\bX^*) \bU_\Theta$, and thus $\mathbb{E}[\bU_\bY] = \Psi(\bX^*) \mathbb{E}[\bU_\Theta]$ and $\var[\bU_\bY] = \Psi(\bX^*) \var[\bU_\Theta] \Psi(\bX^*)^\T$. The matrix $\bU_\Theta$ represents the discrepancy of the spatial spline coefficients for the reality model. Should our MME contain models with high spatial resolution we could specify $\bU_\Theta$ as a lower-rank process \citep[e.g. with fixed-rank methods such as][]{cressie2008fixed}; we do not find the need to do so here. We set $\mathbb{E}[\bU_\Theta] = 0$ and $\var[\bU_\Theta] = \bI_{l \times l} \otimes \var[\bU_{\bY_{\mS}}]$ where $l$ is the number of spline coefficients at each location and $\var[\bU_{\bY_{\mS}}]$ is calculated by the $C^4$-Wendland function in (\ref{eqn:wendland}) with parameters $\kappa = 0.3$, $c = 4$, and $\tau = 6$; again, we include a sensitivity analysis in the supplementary material to examine the sensitivity of our inference to differing parameterisations. The first update $\mathbb{E}_{\hat{\bbeta}}[\bY^*]$ and $\var_{\hat{\bbeta}}[\bY^*]$ is calculated given $\bX^*$ using $\mathbb{E}_{\hat{\bbeta}}[\mathcal{M}_{\bbeta}]$ and $\var_{\hat{\bbeta}}[\mathcal{M}_{\bbeta}]$ as previously calculated. The second update $\mathbb{E}_{\hat{\bbeta}, \bZ_\bY}[\bY^*]$ and $\var_{\hat{\bbeta}, \bZ_\bY}[\bY^*]$ are calculated given $\bX^*$ as in (S6) and (S7). Expectations and marginal standard deviations of both updates, calculated with $\bX^* = \mathbb{E}_{\bar{\bX}, \bZ_\bX}[\bX^*]$, are shown in Figure~\ref{fig:ice}. The expected SSTs do not produce adequate sea-ice when only considering learnt relations from the MME (Figure~\ref{fig:ice}a) and marginal standard deviation of SIC is large in regions where the expected SST is cold enough to guarantee sea-ice coverage (Figure~\ref{fig:ice}b). Updating SIC using the data leads to more extensive sea-ice cover (Figure~\ref{fig:ice}c) and a reduction of standard deviation everywhere except the sea-ice boundary (Figure~\ref{fig:ice}d). We also show, in Figures~\ref{fig:ice}e and \ref{fig:ice}f the empirical ensemble mean and standard deviation, respectively. Our final reconstruction produces more sea-ice in the Southern Hemisphere, but more crucially, drastically reduces the sea-ice uncertainty in the sea-ice interior. Similar behaviour is seen in the Northern Hemisphere winter.

\spacingset{1}

\begin{figure}[h!]
  \centering
  \includegraphics[width = 140mm]{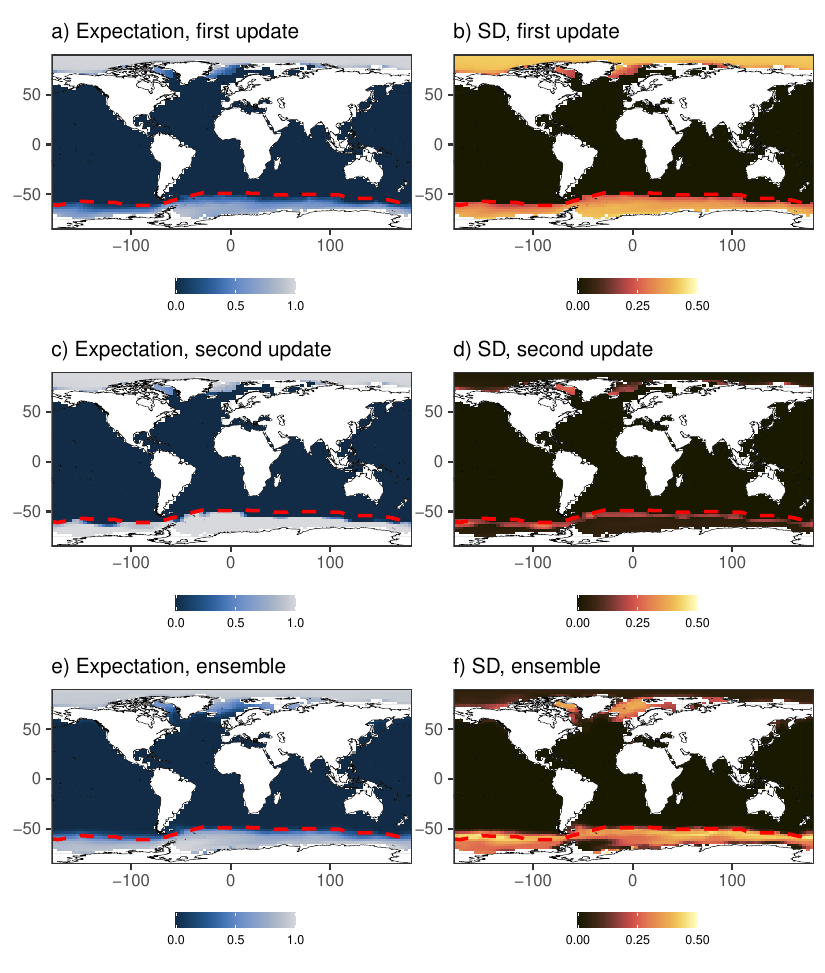}
  \caption{Adjusted beliefs of August SIC given expected SST, $\bX^* = \mathbb{E}_{\bar{\bX}, \bZ_\bX}[\bX^*]$: (a) expectation of SIC adjusted by $\hat{\bbeta}$; (b) marginal standard deviation of SIC adjusted by $\hat{\bbeta}$; (c) expectation of SIC adjusted by $\hat{\bbeta}$ and $\bZ_\bY$; (d) marginal standard deviation of SIC adjusted by $\hat{\bbeta}$ and $\bZ_\bY$; (e) ensemble mean; and (f), marginal standard deviation of the ensemble. Values are of sea-ice concentration measured between 0 and 1.}
  \label{fig:ice}
\end{figure}

\spacingset{\spacing}

\spacingset{1} 
\section{Sampled boundary conditions and their influence on glacial ice sheet modelling outputs} \label{sec:influence}
\spacingset{\spacing}

The remit of this work was to reconstruct, with uncertainty, joint SST and SIC fields to act as boundary conditions into FAMOUS-Ice \citep{smith2021famous}, an `atmosphere only' global climate model coupled to a ice sheet model. To examine the effects that the reconstruction has on the atmosphere-ice sheet model outputs we run a small ensemble varying only the SST and SIC boundary conditions. Bayes linear analysis updates our beliefs of the first two moments of $\bX^*$ and $\bY^*$. Define the Cholesky decompositions $\bL_{\bX}$ and $\bL_{\bY}$, so that $\bL_{\bX} \bL_{\bX}^\T = \var_{\bar{\bX}, \bZ_\bX}[\bX^*]$ and $\bL_{\bY} \bL_{\bY}^\T = \var_{\hat{\bbeta}, \bZ_\bY}[\bY^*]$. We may probabilistically sample 
  \begin{linenomath}
  \begin{align}
      \tilde{\bX}^* &\sim \mathbb{E}_{\bar{\bX}, \bZ_\bX}[\bX^*] + \bL_{\bX} \bZ \label{eqn:cheb1} \\
      \tilde{\bY}^* &\sim \mathbb{E}_{\hat{\bbeta}, \bZ_\bY}[\bY^*] + \bL_{\bY} \bZ \label{eqn:cheb2}
  \end{align} 
  \end{linenomath}
where $\bZ$ is a vector of independent random variables $\mZ_i$ with $\mathbb{E}[\mZ_i] = 0$ and $\var[\mZ_i] = 1$. Note, assigning a distributional form to the $\mZ_i$ is a further choice; for example, we could assign a Gaussian, uniform, or any other distribution that would make sense given the context. Similar to how ensemble design is considered in history matching \citep[e.g.][]{salter2019uncertainty}, we may eschew such distribution assumptions and set the bounds of the $\mZ_i$ with an appeal to Chebyshev's inequality. As is standard in the history matching literature, we set the concentration parameter $k = 3$ and sample $\mZ_i \sim \mathcal{U}(-k,k)$; we call this the \textit{plausible} set, and stress that it is not a probabilistic design, rather, a bounding notion of plausibility. To generate a joint sample ($\tilde{\bX}^*$, $\tilde{\bY}^*$), we first sample $\tilde{\bX}^*$ and then $\tilde{\bY}^*$. The structural dependencies between the sampled $\tilde{\bX}^*$ and $\tilde{\bY}^*$ are captured by our updated beliefs on $\mathcal{M}_{\bbeta}$ and $\bU_\Theta$. An example of a joint sample $(\tilde{\bX}^*, \tilde{\bY}^*)$ for the months of February and August is given in Figure~\ref{fig:hm}.

\spacingset{1} 

\begin{figure}[h!]
  \centering
  \includegraphics[width = 140mm]{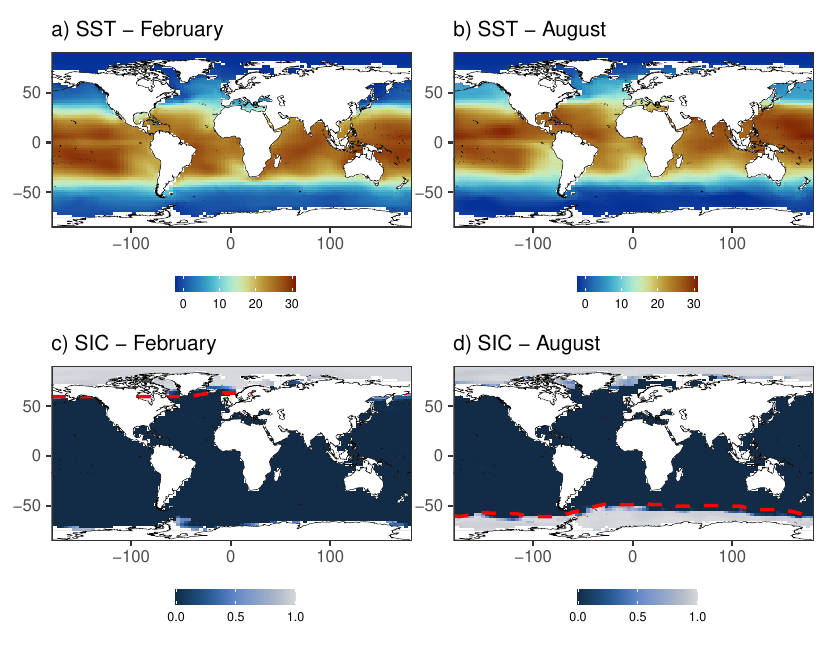}
  \caption{A single plausible sample from the SST and SIC reconstructions. Figures~\ref{fig:hm}a and \ref{fig:hm}b show SSTs from February and August, respectively and Figures~\ref{fig:hm}c and \ref{fig:hm}d show SICs from February and August, respectively.}
  \label{fig:hm}
\end{figure}

\spacingset{\spacing}

We generate an ensemble of 25 runs comprising a single reference run utilising the mean SST and SIC fields produced as part of the PMIP4 LGM experiments, and 24 randomly generated plausible samples of SST and SIC. Averaging over PMIP models to act as boundary conditions in ice-sheet modelling is commonplace \citep[e.g.][]{kageyama2017pmip4}. Crucially, it should be noted that the reference run and each of the PMIP runs $(\mathbb{X}, \mathbb{Y})$ do not lie within the calculated plausibility bounds. To examine the impact of running with plausible boundary conditions, we compare reference and plausible ice sheet heights at four geographically distinct locations: Arctic Canada, Central Greenland, Hudson Bay, and the Pacific coast; this is shown in Figure~\ref{fig:glacier_runs}. The simulator is run for 5000 years, with the ice sheet initialised with the LGM Glac-1D reconstruction \citep{tarasov2012data}. Beyond the SST and SIC fields, no other model parameters were changed and the model set-up was based on previous simulations of the Greenland ice sheet \citep{gregory2020large}.

\spacingset{1}

\begin{figure}[h!]
\centering
 \begin{subfigure}[b]{\textwidth}
     \centering
     \includegraphics[width = 120mm]{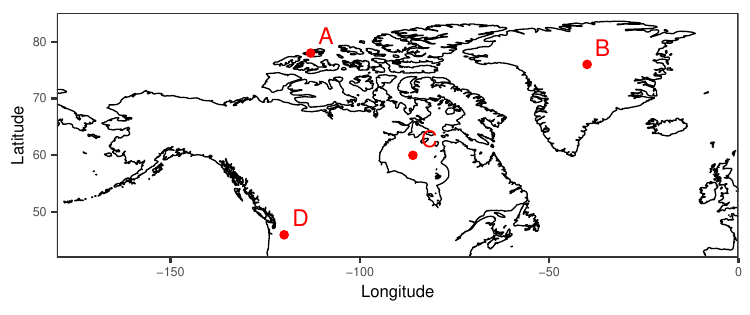}
 \end{subfigure}
 \begin{subfigure}[b]{\textwidth}
     \centering
     \includegraphics[width = 160mm]{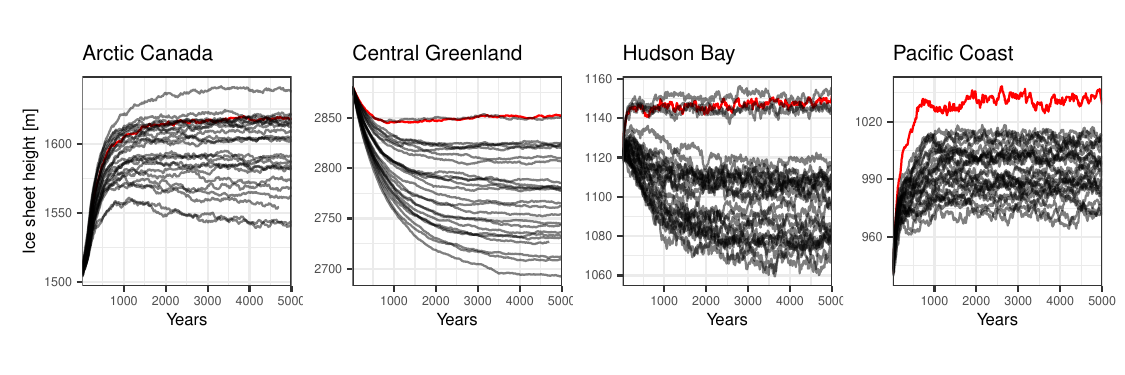}
 \end{subfigure}
 \caption{Ensemble time-series plots at four spatial locations: Arctic Canada (A), Central Greenland (B), Hudson Bay (C), and the Pacific Coast (D). The reference run is shown in red, the simulations forced with sampled boundary conditions in black.}
 \label{fig:glacier_runs}
\end{figure}

\spacingset{\spacing}

Figure \ref{fig:glacier_runs} shows that the boundary conditions have a strong influence on the simulated ice sheet. SST and SIC can affect ice sheets either through changing the evaporation over the ocean that transforms into snow falling on the ice sheet, or by warming/cooling the regions close to the oceans thus affecting the surface melt rate. We find that the primary differences in ice sheet size are due to changes in snow accumulation. Our ensemble mainly produces lower ice elevation than the reference run. Examination of the individual ensemble members revealed that this was due to the cooler Eastern Pacific and Western Atlantic boundary conditions reducing evaporation over these oceans causing a reduction in the accumulation of snow onto the ice sheet. The difference between the reference run and the samples is most pronounced at the Pacific coast. This matches our expectation that ocean-proximal sites are more sensitive to marine influence of the SST and SIC than more continental sites. This is also where we see the strongest reduction in SST in our samples compared to the PMIP model simulations. Indeed, the second update causes a strong Pacific cooling along the coast of North America (Figure 2e) relative to what we expect based on the PMIP models. This is a region where models tend to underestimate the upwelling of cold waters from the deep ocean.

Such biases are common in climate models, but are particularly problematic for coupled climate-ice sheet models where the strong feedbacks between climate and ice sheets amplify the effects of climate biases, which can lead to runaway ice sheets (amplified growth or decay) and unrealistic geometries. The ensemble shows a substantial spread of ice elevation (5-10\% of height) caused by the variance in reconstructed SST and SIC, highlighting the importance for considering this source of uncertainty for modelling past ice sheets. Our results show that correcting for biases and incorporating uncertainty in surface ocean conditions has a substantial effect on the simulated ice sheet, which in turn influences the internal dynamics of the ice sheet and its vulnerability to collapse or propensity to grow. The ice sheet geometry itself has direct impact on atmospheric circulation and an indirect influence on ocean circulation from runoff, thus directly impacting global heat distributions and surface climate conditions. Our methodology provides a way to reduce climate biases by prescribing ocean surface conditions compatible with observations, while at the same time exploring the effects of this source of uncertainty on other parts of the earth system. 

\spacingset{1}
\section{Discussion} \label{sec:conclusion}
\spacingset{\spacing}

By exploiting natural conditional exchangeability judgements we develop theory for the coexchangeable process model, as an extension to \cite{rougier2013second}, that combines multi-model ensembles and data to model correlated spatio-temporal processes. We provide results for efficient and scalable inference that may be used for large spatio-temporal problems where probabilistic Bayesian methods are often computationally infeasible \citep[see, for example, ][]{sansom2021constraining}. Our methodology requires fewer assumptions and less onerous belief specifications than that required by a probabilistic Bayesian analysis. To achieve these advancements we develop a Bayes linear analogue to a hierarchical Bayesian model. By combining exchangeability judgements and using the reparameterisation of \cite{hodges1998some}, we extend the Bayes linear exchangeable regressions methodology. We obtain hitherto missing desirable properties present in traditional Bayesian hierarchical models such as the ability to make individual group level inference.

Large scale computational models often have complex spatio-temporal boundary conditions. This is particularly true for Earth system modelling, when any simulation of part of the Earth system, requires other spatio-temporal fields as boundary conditions. Our application looked at palaeo-era ice-sheet modelling, where our model had a coupled atmosphere and ice-sheet, with the SST and SIC as prescribed boundary conditions. These are usually specified using results from a reference simulation, or using a member of a Model Intercomparison Project (MIP). However, individual simulations of Earth system components are known to have biases and any individual simulation cannot adequately represent uncertainty due to boundary conditions. An idea for future investigation is to use the differences between MIP phases to estimate the ensemble discrepancy. This would allow for differences between older and newer MIP phases to inform the discrepancy between the newer models and reality. Whilst outside the scope of this work, using MIP iterations to inform model discrepancy is an interesting avenue, especially for models of the present-day where lots of data are available for validation.

Our methodology allows MIP simulations to be combined with observations efficiently, leading to joint reconstructions of climate boundary conditions that can be used in any area of Earth System modelling. We demonstrate its efficacy by reconstructing last glacial maximum SST and SIC to force an ice-sheet model. We show that the differences between reference ice-sheets and ice-sheets under plausible boundary conditions were considerable and that the uncertainty in the ice-sheet due to propagated boundary condition uncertainty is not ignorable. Other aspects of the Earth system are likely to be sensitive to their boundary conditions, so that joint reconstructions of the type we present here would allow MIP simulations and data to be combined in order to correct existing biases and quantify an important source of uncertainty. The use of MIP ensembles to drive simulations of Earth system components leads to important questions around how these ensembles should be designed. Our method makes the case that a priori exchangeability across as many models as possible is an important design goal.

\spacingset{1}
\section{Acknowledgements}
\spacingset{\spacing}

\if1\blind
{ All authors were funded by UKRI Future Leaders Fellowship MR/S016961/1. Climate-ice sheet simulations were undertaken on ARC4, part of the High Performance Computing Facilities at the University of Leeds.} \fi
\if0\blind
{ \textit{Acknowledgements have been redacted for blind review.}} \fi

\bibliographystyle{model2-names.bst}
\bibliography{./paper}

\end{document}